\documentclass{article}

\usepackage[usenames,dvipsnames]{xcolor}
\usepackage{todonotes}
\usepackage{comment}

\usepackage{microtype,graphicx,subfigure,booktabs}
\usepackage{mathtools,amsthm,amsmath,amssymb,natbib,textcomp}

\usepackage{algorithm2e}
\usepackage{algorithmic}
\usepackage{fullpage}

\usepackage{hyperref}
\hypersetup{colorlinks=true,linkcolor=blue,filecolor=magenta,urlcolor=cyan,}

\title{Regret Bounds for Batched Bandits}
\author{Hossein Esfandiari \\ Google Research \and Amin Karbasi \\ Yale University \and Abbas Mehrabian
\\ McGill University \and Vahab Mirrokni \\ Google Research  }

\DeclareMathOperator*{\argmax}{argmax}

\newcommand{\muhat}{\widehat{\mu}}
\newcommand{\thetahat}{\widehat{\theta}}
\newcommand{\thetastar}{\theta^{\star}}
\newcommand{\astar}{a^{\star}}
\newcommand{\acal}{\mathcal{A}}
\newcommand{\R}{\mathbf{R}}
\newcommand{\dop}[2]{\left\langle #1, #2 \right\rangle}
\newcommand{\transpose}{^{\mathsf{T}}}

\theoremstyle{plain}
\newtheorem{theorem}{Theorem}[section]
\newtheorem{lemma}[theorem]{Lemma}

\theoremstyle{definition}

\theoremstyle{remark}
\newtheorem*{remark}{Remark}

\newcommand{\eps}{\varepsilon}
\renewcommand{\epsilon}{\varepsilon}
\newcommand{\E}{\mathbf{E}}

\begin{document}\sloppy
	
\maketitle




\begin{abstract}
We present simple and efficient algorithms for the batched stochastic multi-armed bandit 
and batched stochastic linear bandit problems. We prove  bounds for their expected regrets that improve over the best known regret bounds for any number of batches.
In particular, our algorithms in both settings achieve the optimal expected regrets by using only a logarithmic number of batches. 
We also study the batched adversarial multi-armed bandit problem for the first time and find the optimal regret, up to logarithmic factors, of any algorithm with predetermined batch sizes.
\end{abstract}
\section{Introduction}
A central challenge in  optimizing many complex systems, such as experimental design~\citep{robbins1952some}, clinical trials~\citep{batchedtwoarms}, hyperparameter tuning~\citep{snoek2015scalable}, and product marketing~\citep{bertsimas2007learning}, is to simultaneously explore the unknown parameter space while at the same time exploit the acquired knowledge for maximizing the utility. In theory, the most convenient way is to explore one parameter at a time.  However, in practice, it is often possible/desirable, and sometimes the only way,  to explore several parameters in parallel. A notable example is designing clinical trials, where it is impractical to wait to observe the effect of a drug on a single patient before deciding about the next trial. Instead, groups of patients with multiple treatments are studied in parallel. Similarly, in marketing and advertising, the efficacy of a strategy is not tested on individual subjects one at a time; instead, multiple strategies are run simultaneously in order to gather information in a timely fashion. Similar issues arise in crowdsourcing platforms, where multiple tasks are distributed among users~\citep{kittur2008crowdsourcing}, and time-consuming numerical simulations, which are prevalent in reinforcement learning~\citep{le2019batch,lange2012batch}.

Parallelizing the exploration of the parameter space has a clear advantage: more information can be gathered at a shorter period of time. It also has a clear disadvantage: information cannot be immediately shared across different parallel paths, thus future decisions cannot benefit from the intermediate results. Note that a fully sequential policy, in which exploration is done in a fully sequential manner, and a fully parallel policy, in which exploration is completely specified a priori without any information exchange, are two extremes of the policy spectrum. In fact, carefully calibrating between the information parallelization phase (how many experiments should be run in parallel) and the information exchange phase (how often information should be shared) is crucial for applications in which running experiments is costly or time-consuming. \emph{Batch policies}, which are the focus of this paper, aim to find the sweet spot between these two phases. The main challenge is to carefully design batches of experiments, out of a combinatorially large set of possibilities, which can be run in parallel and explore the parameter space efficiently while being able to exploit the parameter region with the highest utility. 

In this paper, we study the problem of  batch policies in the context of multi-armed  and  linear bandits with the goal of minimizing \emph{regret}, the standard benchmark for comparing  performance of bandit policies. We provide a comprehensive theoretical understanding of these  problems by designing  algorithms along with hardness results. 


\section{Bandits, Regret, and Batch Policies}
A \emph{bandit problem} is a sequential game between a player/learner and an environment. The game is played over $T$ rounds, called the {\em time horizon}. In each round, first the player chooses an action from a set of actions and then the environment presents a reward. For instance, in clinical trials, the actions correspond to the available treatments and the rewards correspond to whether the treatment cured the patient or not. 
As the player does not know the future, she follows a \emph{policy}, a mapping from histories to actions. Similarly, the environment can  also be formulated as a mapping from actions to rewards. Note that both the player and the environment may randomize their decisions. The standard performance measure for a bandit policy is its \emph{regret}, defined as the difference between the total expected reward collected by the policy and the total expected reward collected by the optimal policy. In this paper, we focus on the following bandit models: stochastic and adversarial multi-armed bandits and 
stochastic linear bandits.

\subsection{Multi-Armed Bandits}
Traditionally, actions are referred to as `arms'
and `taking an action' is referred to as `pulling an arm.'
A \emph{multi-armed bandit} is a one-player game in which the number of arms is a finite number $K$. 
Let $[K]\coloneqq\{1,2,\dots,K\}$ denote the set of arms.
In each round $t=1,2,\dots,T$, the player pulls an arm $a_t\in[K]$ and receives a corresponding reward $r_t$. 

We consider two possible generation models for the rewards:
the \emph{stochastic} setting and the \emph{adversarial} setting.
In the former, the  rewards of each arm are sampled in each round independently from some fixed distribution supported on $[0,1]$.
In other words, each arm has a potentially different reward distribution, while the distribution of each arm does not change over time.
Suppose the player pulls the arms
$a_1,a_2,\dots,a_T$, and 
suppose $\mu_i$ denotes the mean of arm $i\in[K]$.
Then, in the stochastic setting, the \emph{expected regret} is defined as
\[
{\E[\textnormal{Regret}]}\coloneqq T \max_{i\in[K]}\mu_i - \E\left[\sum_{t=1}^{T} \mu_{a_t} \right].
\]

In contrast to the stochastic setting, a multi-armed adversarial bandit is specified by an \emph{arbitrary} sequence of rewards $(r_{i,t})_{i\in[K],t\in[T]}\in[0,1]$. In each round $t$, the player chooses a distribution $P_t$; an arm $a_t\in[K]$ is sampled from $P_t$ and the player receives reward $r_{a_t,t}$. In this sense, a policy can be viewed as a function mapping history sequences to distributions over arms. The  expected regret is defined as the difference between the expected reward collected by the policy and  the best fixed action in hindsight, i.e., 
\[
\E[\textnormal{Regret}]\coloneqq \max_{i\in[K]} \sum_{t=1}^T r_{i,t} - \E\left[\sum_{t=1}^T r_{a_t,t}\right]. 
\]
Note that the only source of randomness in the regret stems from the  randomized policy used by the player. 
Randomization is indeed crucial for the player and is the only way for her to avoid a regret of $\Omega(T)$.

\begin{remark}
The assumption that the rewards are bounded in $[0,1]$ is a standard normalization assumption in online learning, but it is immediate to generalize our algorithms and analyses to reward distributions bounded in any known interval, or
(in the stochastic case) to Gaussian or subgaussian distributions whose mean lie in a known interval. 
The only change in the analysis is that instead of using Hoeffding's inequality which requires a bounded distribution, one has to use 
a concentration inequality for sums of subgaussian distributions, see, e.g., \citet[Proposition~2.5]{wainwright2019high}.
\end{remark}

\subsection{Stochastic Linear Bandits}
In a stochastic linear bandit, each arm is a vector $a\in \R^d$ belonging to some action set $\mathcal{A}\subseteq \R^d$, and there is a parameter $\thetastar\in \R^d$ unknown to the player. In round $t$, the player chooses some action $a_t\in \mathcal{A}$ and receives reward $r_t = \langle a_t , \thetastar \rangle + \mu_t$, where  $\mu_t$ is   a zero-mean 1-subgaussian noise;
that is, $\mu_t$ is independent of other random variables,
has $\E \mu_t = 0$ and satisfies
\[
\E \left[ e^{\lambda \mu_t} \right] \leq \exp (\lambda^2/2)
\qquad \forall \lambda \in \R.
\]
Note that any zero-mean Gaussian distribution with variance at most 1
and any zero-mean distribution supported on an interval of length at most 2 satisfies the above inequality.
For normalization purposes, we assume that 
$\|\thetastar\|_2\leq1$
and
$\|a\|_2\leq1$ for all arms $a\in \mathcal{A}$.

Denoting the pulled arms by $a_1,\dots,a_T$, the expected regret of a stochastic linear bandit algorithm is defined as
\[
\E[\textnormal{Regret}]
\coloneqq T \sup_{a\in\acal}\dop{a}{\thetastar} -\E\left[ \sum_{t=1}^{T} \dop{a_t}{\thetastar}\right].
\]

\subsection{Batch Policies}
As opposed to the bandit problems described above, in the \emph{batch mode}, the player commits to a sequence of actions (a {\em batch} of actions) and observes the rewards {\em after all actions in that sequence are played}. 
More formally, at the beginning of each batch $i=1,2,\dots$, the player announces a list of arms/actions to be pulled/played.
Afterwards, she receives a list of pairs consisting of arm indices and rewards, corresponding to the rewards generated from these pulls.
Then the player decides about the next batch.

The batch sizes could be chosen non-adaptively or adaptively. In a \emph{non-adaptive} policy, the batch \emph{sizes} are fixed before starting the game, while an \emph{adaptive} policy is one in which the batch sizes may depend on previous observations of the algorithm. Obviously, an adaptive policy is more powerful and may achieve a smaller regret. 
In both cases, the player is subject to using at most a given number of batches, $B$. 
Moreover, the total number of actions played by the player must sum to the {\em horizon} $T$.
We assume that the player knows the values of $B$ and $T$.
Notice that the case $B=T$ corresponds to original bandit problems where actions are committed fully sequentially and has been studied extensively, see, e.g., \citet{torcsababook}.
Thus, we refer to the case $B=T$ as the \emph{original} or the \emph{sequential} setting.

Our algorithms for stochastic bandits are adaptive,
while in the adversarial setting we focus mostly on non-adaptive algorithms.

\section{Contributions and Paper Outline}
We provide analytic regret bounds for the batched version of three bandit problems: stochastic multi-armed bandits, stochastic linear bandits, and adversarial multi-armed bandits.

Recall that $K$ denotes the number of arms,
$T$ the time horizon and $B$ the number of batches.
The case $B=T$ corresponds to the sequential setting which has been studied extensively, while 
if $B=1$ then no learning can happen,
thus we are mostly interested in the regime $1< B < T$.

\subsection{Stochastic Multi-Armed Bandits}
Let $\Delta_i\coloneqq \max_{a\in[K]}\mu_a - \mu_i \geq0$ denote the gap of arm $i$.
For stochastic multi-armed bandits, the optimum regret achievable in the easier sequential setting is
$ O\Big(\log(T)\sum_{i:\Delta_i>0}{\Delta_i}^{-1}\Big)$; this is achieved, e.g., by the well-known upper confidence bound (UCB) algorithm of~\citet{ucb}.

Our first contribution is a simple and efficient algorithm
(Algorithm~\ref{alg:mab})
whose regret scales as
\(O\left(T^{1/B} \log(T)\sum_{i:\Delta_i>0}{\Delta_i}^{-1}\right)\) (Theorem~\ref{thm:mainmab}).
In particular, as soon as $B \geq C \log T$ for some absolute constant $C$, 
it matches the optimum regret achievable in the fully sequential setting.
In other words, increasing the number of batches from $C \log T$ to $T$ does not reduce the expected regret by more than a constant multiplicative factor.
\citet[Corollary~2]{batchedneurips2019} show that 
$B = \Omega(\log T / \log \log T)$ batches are \emph{necessary} to achieve
the optimal regret
$O\Big(\log(T) \displaystyle\sum_{i:\Delta_i>0}\Delta_i^{-1}\Big)$. This lower bound implies that our algorithm uses almost the minimum number of batches needed (i.e., $O(\log T)$ versus $\Omega({\log T}/{\log \log T})$) to achieve the optimal regret.

Our algorithm is based on arm elimination:
in each batch, we explore several arms in parallel,
then at the end of the batch, we eliminate those arms that are certainly suboptimal, and then we repeat with a larger batch.


\subsection{Stochastic Linear Bandits}
In the sequential setting of stochastic linear bandits, the best-known upper bound for regret is $O(d\sqrt{T\log(T)})$
\citep[Note~2 in Section~22.1]{torcsababook},
while the best known lower bound is $\Omega(d\sqrt{T})$~\citep{linbanditlowerbound}.
For a finite number of arms $K$,
the best known upper bound is $O(\sqrt{d T \log K})$
\citep{linbanditsfinitearms}.

Our second contribution is the first 
algorithm for batched stochastic linear bandits (Algorithm~\ref{alg:lb}): an efficient algorithm for the case of finitely many arms with regret
$O\left( T^{1/B} \sqrt{d T \log (KT)}\right)$ (Theorem~\ref{thm:lbmain}), nearly matching the $O(\sqrt{dT\log(K)})$ upper bound for the sequential setting as soon as the number of batches is $\Omega(\log T)$.
If there are infinitely many arms, we achieve a regret upper bound of 
$O\left( T^{1/B} \cdot d \sqrt{T \log (T)}\right)$,
which matches the best-known upper bound of $O(d\sqrt{T \log T})$ for the sequential setting as soon as the number of batches is $\Omega(\log T)$.

Our algorithm for batched stochastic linear bandits is also based on arm eliminations. However, it would blow up the regret if we were to pull each arm in each batch the same number of times; instead we use the geometry of the action set to carefully choose a sequence of arms in each batch, based on an approximately optimal G-design, that would give us the information needed to gradually eliminate the suboptimal arms.
The extension to the case of infinitely many arms is achieved via a discretization argument.


\subsection{Adversarial Multi-Armed Bandits}
The optimal regret for adversarial multi-armed bandits in the sequential setting is $\Theta(\sqrt{KT})$,
see~\citet{adversarialminimax}. Our third contribution is to
prove that the best achievable regret of any non-adaptive algorithm for batched adversarial multi-armed bandits is
$
\widetilde{\Theta}
\left(
\sqrt{T \left(K + \frac T B\right)}
\right),
$
where the $\widetilde{\Theta}$ allows for logarithmic factors (Theorem~\ref{thm:abmain}).
That is, we prove an optimal (minimax) regret bound, up to logarithmic factors, for non-adaptive algorithms for batched adversarial multi-armed bandits.

Finally, we prove a lower bound of $\Omega(T/B)$ for the regret of {\em any} algorithm, adaptive or non-adaptive,
for batched adversarial multi-armed bandits
(Theorem~\ref{thm:adversariallowerbound1}).
This shows a large contrast with the stochastic version, since there is a polynomial relation between the number of batches and the regret. In particular, one needs at least $\Omega(\sqrt{T/K})$ batches to achieve the optimum regret of $O(\sqrt {TK})$. 

The upper bound for batched adversarial multi-armed bandits is proved via a reduction to the setting of multi-armed bandits with delays,
while the lower bounds are proved by carefully designing hard reward sequences.


\paragraph{Paper Outline}
In the next section we review prior work on the batched bandits model.
Stochastic multi-armed bandits,
stochastic linear bandits,
and adversarial multi-armed bandits
are studied
in Sections~\ref{sec:batched},
~\ref{sec:linear}, and~\ref{sec:adversarial} respectively.
We conclude with a discussion and directions for further research in Section~\ref{sec:conclusion}.

\section{Related Work}
Sequential bandit problems, in particular multi-armed bandits, have  been studied for almost a century. While we cannot do justice to all the work that has been done, let us highlight a few excellent monographs \citep{Bubeck:Survey12,Slivkinssurvey, torcsababook}. In the rest of this section, we  review the results in the batched setting. 

\citet{improveducb} present an algorithm for stochastic multi-armed bandits based on arm elimination.
Even though their algorithm is presented for the sequential setting, it can be turned into an algorithm for the batched setting for $\Omega(\log T)$ batches. More precisely, they prove that the optimal problem-dependent regret of $O\Big(\log(T)\sum_{i:\Delta_i>0}{\Delta_i}^{-1}\Big)$ is achievable as soon as the number of batches is $B=\Omega(\log T)$. Our results, however, are more general and hold for \emph{any} number of batches $B \in \{1,\dots,T\}$.

Similarly, \citet[Theorem~6]{onlinelearningswitchingcosts} provide an algorithm using $B=O(\log \log T)$ batches that achieves problem-independent regret of $O(\sqrt{KT \log K})$ for stochastic multi-armed bandits. (Their algorithm is also for the sequential version but it can be batched easily.) This rate is minimax optimal for the sequential version up to a $\sqrt{\log T}$ factor.
However, the focus of our paper for stochastic multi-armed bandits is on problem-dependent bounds, which have logarithmic dependence on $T$.

A special case of batched multi-armed bandits was studied by \citet{batchedtwoarms}, where they only consider two arms and provide upper and lower bounds on the regret. 
In particular, denoting the gap between the arms by $\Delta$,
\citet[Theorem~2]{batchedtwoarms} provide a regret bound of 
\[
\mathbf{E}[
\textnormal{Regret}]\leq O\left(\left(\frac{T}{\log T}\right)^{1/B} \frac{\log T}{\Delta}\right)
\]
when $K=2$.
In contrast, we consider the more general setting of $K\geq2$ and 
we give an algorithm
for \emph{any} number of arms $K$
satisfying
\[
\mathbf{E}[
\textnormal{Regret}]\leq O\left({T}^{1/B}  \sum_{i:\Delta_i>0}\frac{\log(T)}{\Delta_i}
\right).
\]

The closest previous work is 
\citet{batchedneurips2019}, who study the batched stochastic multi-armed bandit in full generality and prove regret bounds.
In particular, \citet[Theorem~4]{batchedneurips2019} provide a non-adaptive algorithm with
\[
\mathbf{E}[
\textnormal{Regret}] \leq 
O \left ( \frac{ K \log(K) \log (KT) T^{1/B}}
{\displaystyle\min_{i:\Delta_i>0} \Delta_i}
\right).
\]
Our Theorem~\ref{thm:mainmab} 
shaves off a factor of $\log (K)$ and also improves the factor $\frac{K}{\displaystyle\min \Delta_i}$ to $\sum_{i}\frac{1}{\Delta_i}$.
This latter improvement may be as large as a multiplicative factor of $K$ in some instances, e.g., if
$\mu_1=1,\mu_2=1-1/K,\mu_3=\dots=\mu_K=0$.
We achieve these improvements by adapting the batch sizes based on the previous outcomes as opposed to predetermined batch sizes.

The other bandit problems we study, namely, stochastic linear bandits and adversarial multi-armed bandits, have not been studied in the batch mode setting prior to our work.

Optimization in batch mode  has also been studied in other machine learning settings where information-parallelization is very effective. Examples include best arm identification \citep{batchedbestarmidentification,batchedbestarmidentification2}, bandit Gaussian processes \cite{desautels2014parallelizing, kathuria2016batched, contal2013parallel}, submodular maximization \cite{balkanski2018adaptive,fahrbach2019submodular, chen2019unconstrained}, stochastic sequential optimization \cite{esfandiari2019adaptivity, agarwal2019stochastic}, active learning \cite{hoi2006batch, chen2013near}, and reinforcement learning \cite{ernst2005tree}, to name a few.




\section{Batched Stochastic Multi-Armed Bandits}
\label{sec:batched}
\subsection{The Algorithm}
The algorithm works by gradually eliminating suboptimal arms.
Let $\delta \coloneqq 1/(2K T B)$ and $q\coloneqq  T^{1/B} $, and define $c_i \coloneqq \lfloor q^1\rfloor + \dots + \lfloor q^i\rfloor$.
Note that $c_B \geq T$.
Initially, all arms are `active.'
In each batch $i=1,2,\dots$, except for the last batch, each active arm is pulled $\lfloor q^i \rfloor$ times.
Then, after the rewards of this batch are observed, the mean of each active arm is estimated as the average reward received from its pulls. An arm is then eliminated if its estimated mean is smaller, by at least $\sqrt {{2 \ln(1/\delta)}/{c_i}}$, than the estimated mean of another active arm.
The last batch is special: if we have used $i-1$ batches so far and the number of active arms times $\lfloor q^i \rfloor$ exceeds the number of remaining rounds, the size of the next batch equals the number of remaining rounds, and in this last batch we pull the active arm with the largest empirical mean.
See Algorithm~\ref{alg:mab} for the pseudocode.

\begin{algorithm}
   \caption{Batched Arm Elimination for Stochastic Multi-Armed Bandits}
   \label{alg:mab}
\begin{algorithmic}
   \STATE {\bfseries Input:} number of arms $K$, time horizon $T$, number of batches $B$
   \STATE $q \longleftarrow T^{1/B}$
   \STATE $\acal \longleftarrow [K]$ \COMMENT{active arms}
   \STATE $\muhat_a \longleftarrow 0$ for all $a\in\acal$ \COMMENT{estimated means}
   \FOR{$i=1$ {\bfseries to} $B-1$}
   \IF{$\lfloor q^i \rfloor \times |\acal|>$ remaining rounds}
   \STATE {\bfseries break}
   \ENDIF
   \STATE {In the $i$th batch, play each arm $a\in\acal$ for $\lfloor q^i \rfloor$ times} 
   \STATE{Update $\muhat_a$ for all $a\in\acal$}
   \STATE {$c_i\longleftarrow \sum_{j=1}^{i} \lfloor q^j \rfloor$}
   \FOR{$a\in \acal$}
   \IF{$\muhat_a < \max_{a\in \acal}\muhat_a-\sqrt{2\ln(2KTB)/c_i}$}
   \STATE{Remove $a$ from $\acal$}
   \ENDIF
   \ENDFOR
   \ENDFOR
   \STATE{In the last batch, play $\argmax_{a\in\acal}\muhat_a$}
\end{algorithmic}
\end{algorithm}

\subsection{Analysis}

\begin{theorem}
\label{thm:mainmab}
The expected regret of Algorithm~\ref{alg:mab} for batched stochastic multi-armed bandits is bounded by 
\[
\mathbf{E}[
\textnormal{Regret}] \leq 9  T^{1/B} 
\ln(2KTB)
\sum_{j:\Delta_j>0}\frac{1}{\Delta_j}.
\]
\end{theorem}

\begin{proof}
We recall  Hoeffding's inequality.

\paragraph{ 
\protect{\citet[Theorem~2]{hoeffding}}.}
Suppose $X_1,\dots,X_n$ are independent, identically distributed random variables supported on $[0,1]$. Then, for any $t\geq0$,
\[
\mathbf{Pr} \left[\left|\frac 1 n \sum_{i=1}^{n} X_i - \mathbf{E}X_1 > t\right|\right] < 2 \exp (-2nt^2).
\]

For an active arm at the end of some batch $i$, we say its estimation is `correct' if the estimation of its mean is within
$\sqrt{{\ln(1/\delta)}/{2c_i}}$
of its actual mean.
Since each active arm is pulled $c_i$ times by the end of batch $i$, by Hoeffding's inequality, 
the estimation of any active arm at the end of any  batch (except possibly the last batch) is correct with probability at least $1-2\delta$.
Note that there are $K$ arms and at most $B$ batches (since $c_B\geq T$). Hence, by the union bound, the probability that the estimation is incorrect for some arm at the end of some batch is bounded by $K B \times 2\delta = 1/T$.
Whenever some estimation is incorrect, we upper bound the regret by $T$. Hence, 
the total expected regret can be bounded as
\[\mathbf{E}[
\textnormal{Regret}]
\leq
\frac 1 T \times T + \mathbf{E}[
\textnormal{Regret} | \textnormal{estimations are correct}].
\]
So, from now on, we assume that all estimations are correct, which implies that each {\em gap} is correctly estimated within an additive factor  of $\sqrt{{2\ln(1/\delta)}/{c_i}}$. Thus, the best arm is never eliminated.
We can then write the expected regret as
\begin{align*}
\mathbf{E}&[
\textnormal{Regret} | \textnormal{estimations are correct}] \\&= \sum_{j:\Delta_j>0} \Delta_j \mathbf{E} [T_j | \textnormal{estimations are correct}], 
\end{align*}
where $T_j$ denotes the total number of pulls of arm $j$.
Let us now fix some (suboptimal) arm $j$ with $\Delta_j>0$ and upper bound $T_j$.
Suppose batch $i+1$ is the last batch that arm $j$ was active.
Since this arm is not eliminated at the end of batch $i$, and the estimations are correct, we have
$\Delta_j \leq 2 \sqrt{2 \ln(1/\delta)/c_i}$ deterministically, which means
$c_i \leq 8\ln(1/\delta)\Delta_j^{-2}$.
The total pulls of this arm is thus
\[
T_j \leq c_{i+1} = q + q c_i \leq q + \frac{8 q \ln(1/\delta)}{ \Delta_j^2}, 
\]
whence,
\begin{align*}
\mathbf{E} [
\textnormal{Regret} | & \textnormal{estimations are correct}]
 \\&\leq \sum_{j:\Delta_j>0} \left\{
 q \Delta_j + \frac{8q\ln(1/\delta)}{\Delta_j}
 \right\}\\&
 \leq 
 q (K-1) + 8 q \ln(1/\delta)
 \sum_{j:\Delta_j>0}\frac{1}{\Delta_j}.
\end{align*}
Substituting the values of $q$ and $\delta$ completes the proof.
\end{proof}

\begin{remark}
Instead of setting $q=T^{1/B}$, one can choose any value for $q$ that satisfies $\sum_{i=1}^{B}\lfloor q^i \rfloor \geq T$ 
and the analysis would not change, while the $T^{1/B}$ in the regret bound can be replaced with $q$.
\end{remark}




\section{Batched Stochastic Linear Bandits}
\label{sec:linear}


\subsection{The Algorithm}
The algorithm is based on arm elimination, as in the multi-armed case.
Here is the key lemma, which follows from the results in~\citet[Chapter~21]{torcsababook}.
Note that we may assume without loss of generality that the action set spans $\R^d$, otherwise we may work in the space spanned by the action set.

\begin{lemma}\label{lem:goptimal}
For any finite action set $\acal$ that spans $\R^d$ and any $\delta,\eps>0$, we can efficiently find a \emph{multi-set} of $\Theta(d \log(1/\delta)/\eps^2)$ actions (possibly with repetitions) such that they span $\R^d$ and if we perform them in a batched stochastic linear bandits setting and let $\thetahat$ be the least-squares estimate for $\thetastar$, then, for any $a\in\acal$, with probability at least $1-\delta$ we have
$\left|\dop{a}{\thetahat-\thetastar}\right| \leq \eps$.
\end{lemma}

\begin{remark}
If the performed actions are $a_1,\dots,a_n$ and the received rewards are $r_1,\dots,r_n$,
then the least squares estimate for $\thetastar$ is
\[
\thetahat \coloneqq
\left(\sum_{i=1}^{n} a_i a_i\transpose \right)^{-1}
\left(\sum_{i=1}^{n} r_i a_i\right).
\]
\end{remark}

\begin{remark}
It is known that a multi-set of actions with the guarantee of Lemma~\ref{lem:goptimal} exists and has size at most 
$\frac{d^2+d}{2}+\frac{d\ln(2/\delta)}{\eps^2}$;
it can be defined based on a so-called G-optimal design for $\acal$, see~\citet[equation (21.3)]{torcsababook}, or geometrically, the minimum volume ellipsoid containing $\acal$.
However, it is not clear whether we can find the G-optimal design {\em efficiently};
however, it is known that approximate G-optimal designs of size $\Theta(d \log(1/\delta)/\eps^2)$ can be found efficiently, see~\citet[Note~3 in Section~22.1]{torcsababook}.
\end{remark}

Let $B$ denote the number of batches, and let $c$ and $C$ be the constants hidden in the $\Theta$ notation in Lemma~\ref{lem:goptimal}; namely, the size of the multi-set is in
$[c d \log(1/\delta)/\eps^2,
C d \log(1/\delta)/\eps^2]$.
Define $q \coloneqq (T/c)^{1/B}$ and $\eps_i \coloneqq \sqrt{d \log(KT^2)/q^i}$.

We now describe the algorithm. Initially, all arms are active.
In each batch $i=1,2,\dots$, except for the last batch, we compute the multi-set given by Lemma~\ref{lem:goptimal}, with $\acal$ being the set of active arms, $\delta\coloneqq1/KT^2$ and $\eps=\eps_i$; we then perform the actions given by the lemma, compute $\thetahat_i$, and eliminate any arm $a$ with
\begin{equation}
\label{eq:eliminationrule}
\dop{a}{\thetahat_i}
< \max_{\textnormal{active } a}\dop{a}{\thetahat_i}
- 2\epsilon_i.
\end{equation}
In the last batch, we pull the  active arm with the largest dot product with the last estimated $\thetahat$.
The pseudocode can be found in Algorithm~\ref{alg:lb}.

\begin{algorithm}
   \caption{Batched Arm Elimination for Stochastic Linear Bandits With Finitely Many Arms}
   \label{alg:lb}
\begin{algorithmic}
   \STATE {\bfseries Input:} action set $\acal \subseteq \R^d$, time horizon $T$, number of batches $B$
   \STATE $q \longleftarrow (T/c)^{1/B}$
   \FOR{$i=1$ {\bfseries to} $B-1$}
   \STATE{$\eps_i\longleftarrow \sqrt{d\log(KT^2)/q^i}$}
   \STATE{$a_1,\dots,a_n \longleftarrow$ multi-set given by Lemma~\ref{lem:goptimal} with parameters $\delta=1/KT^2$ and $\eps=\eps_i$}
   \IF{$n>$ remaining rounds}
   \STATE {\bfseries break}
   \ENDIF
   \STATE {In the $i$th batch, play the arms $a_1,\dots,a_n$ and receive rewards $r_1,\dots,r_n$} 
   \STATE{$\thetahat \longleftarrow
\left(\sum_{i=1}^{n} a_i a_i\transpose \right)^{-1}
\left(\sum_{i=1}^{n} r_i a_i\right)$}
   \FOR{$a\in \acal$}
   \IF{$\dop{a}{\thetahat}
< \max_{a\in\acal}\dop{a}{\thetahat}
- 2\epsilon_i$}
   \STATE{Remove $a$ from $\acal$}
   \ENDIF
   \ENDFOR
   \ENDFOR
   \STATE{In the last batch, play $\argmax_{a\in\acal}\dop{a}{\thetahat}$}
\end{algorithmic}
\end{algorithm}

\subsection{Analysis}

\begin{theorem}
\label{thm:lbmain}
The regret of Algorithm~\ref{alg:lb} is at most
$$\mathbf{E}[
\textnormal{Regret}] \leq O\left( T^{1/B} \sqrt{d T \log (KT)}\right)$$
for the batched stochastic linear bandit problem with $K$ arms.
\end{theorem}
\begin{proof}
Let $n_i$ denote the size of batch $i$. Then, by Lemma~\ref{lem:goptimal} we have
$n_i \in [cq^i, Cq^i]$.
Since $q = (T/c)^{1/B}$, the number of batches is not more than $B$.

We define the following \emph{good event}:
``for any arm $a$ that is active in the beginning of batch $i$, 
at the end of this batch we have 
$\left|\dop{a}{\thetahat_i-\thetastar}\right| \leq \eps_i$.''
Since there are $K$ arms and at most $T$ batches and $\delta=1/KT^2$,
by Lemma~\ref{lem:goptimal} and the union bound the good event happens with probability at least $1-1/T$. We assume it happens in the following, for if it does not happen, we can upper bound the regret by $T$, adding just 1 to the final regret bound, as in the proof of Theorem~\ref{thm:mainmab}.

Since the good event happens, and because of our elimination rule~\eqref{eq:eliminationrule}, the triangle inequality shows the optimal arm will not be eliminated: 
let $\astar$ denote the optimal arm;
for any suboptimal arm $a$,
\begin{align*}
\dop{a}{\thetahat_i}-
\dop{\astar}{\thetahat_i}
&\leq
(\dop{a}{\thetastar}+\eps_i)-
(\dop{\astar}{\thetastar} - \eps_i)
 \\&< 2\eps_i.
\end{align*}

Next, fix a suboptimal arm $a$, and let $\Delta\coloneqq \dop{a^{\star}-a}{\thetastar}$ denote its gap.
Let $i$ be the smallest positive integer such that $\eps_i < \Delta/4$.
Then, since the good event happens, 
and because of our elimination rule \eqref{eq:eliminationrule}, the triangle inequality shows this arm will be eliminated by the end of batch $i$:
\begin{align*}
\dop{\astar}{\thetahat_i}-
\dop{a}{\thetahat_i}
& \geq
(\dop{\astar}{\thetastar} - \eps_i)
-
(\dop{a}{\thetastar}+\eps_i)
 \\& = \Delta - 2\eps_i > 2\eps_i.
\end{align*}
The above claims show that during batch $i$, any active arm has gap at most $4\eps_{i-1}$, so the instantaneous regret in any round is not more than $4\eps_{i-1}$, whence the expected regret of the algorithm conditional on the good event can be bounded by:
\begin{align*}
\sum_{i=1}^{B}4 n_i \eps_{i-1}
&\leq
4C
\sum_{i=1}^{B} q^i 
\sqrt{d \log(KT^2)/q^{i-1}}
\\&\leq
6Cq\sqrt{d\log(KT)}
\sum_{i=0}^{B-1}q^{i/2}
\\&=
O\left(q\sqrt{d\log(KT)}
q^{B/2}\right)
\\&= O\left( q \sqrt{d T \log (KT)}\right),
\end{align*}
completing the proof.
\end{proof}






\subsection{Infinite Action Sets}
In this section we prove a regret bound of $O\left( T^{1/B} \cdot d \sqrt{T \log T}\right)$ for batched stochastic linear bandits even if the action set $\acal$ has infinite cardinality.
An \emph{$\epsilon$-net} for $\acal$ is a set $\acal' \subseteq \acal$ such that for any $a\in\acal$ there exists some $a'\in\acal'$ with $\|a-a'\|_2\leq\epsilon$.
Since $\acal$ is a subset of the unit Euclidean ball in $\R^d$, it has a $\frac{1}{T}$-net $\acal'$ of cardinality not more than $(3T)^d$, see, e.g., \citet[Corollary~4.2.13]{vershyninbook}.

We execute Algorithm~\ref{alg:lb} using the finite action set $\acal'$. Let $a_1,\dots,a_T$ denote the algorithms' actions. Then, by Theorem~\ref{thm:lbmain}, we have
\begin{align*}
&T \sup_{a\in\acal'}\dop{a}{\thetastar} -\E\left[ \sum_{t=1}^{T} \dop{a_t}{\thetastar}\right]\\
& = O\left(T^{1/B} \sqrt{dT\log(|\acal'|T)}\right)\\
& = O\left( T^{1/B} \cdot d \sqrt{T \log T}  \right).
\end{align*}
On the other hand, since $\acal'$ is a $\frac{1}{T}$-net for $\acal$, for any $a\in\acal$ there exists some $a'\in\acal'$ with $\|a-a'\|_2\leq\frac{1}{T}$, which implies $$\dop{a}{\thetastar}-\dop{a'}{\thetastar}\leq\|a-a'\|_2 \cdot \|\thetastar\|_2\leq\frac{1}{T},$$
and in particular,
\[
\sup_{a\in\acal}\dop{a}{\thetastar}-
\sup_{a\in\acal'}\dop{a}{\thetastar}
\leq\frac{1}{T},
\]
and thus,
\begin{align*}
\E[\textnormal{Regret}]
& = T \sup_{a\in\acal}\dop{a}{\thetastar} -\E\left[ \sum_{t=1}^{T} \dop{a_t}{\thetastar}\right]\\
& =
\left(T \sup_{a\in\acal}\dop{a}{\thetastar}
-
T \sup_{a\in\acal'}\dop{a}{\thetastar}\right) \\
&+
\left(T \sup_{a\in\acal'}\dop{a}{\thetastar} -\E\left[ \sum_{t=1}^{T} \dop{a_t}{\thetastar}\right]\right)\\
&\leq 1 + O\left( T^{1/B} \cdot d \sqrt{T \log T}  \right),
\end{align*}
proving the claim.


\section{Batched Adversarial Multi-Armed Bandits}
\label{sec:adversarial}

We start by proving a regret upper bound.

\begin{lemma}
\label{lem:adversarialupperbound}
There is a non-adaptive algorithm for batched adversarial multi-armed bandits with regret bounded by
\[
\mathbf{E}[
\textnormal{Regret}] \leq O\left(\sqrt{\left(K+\frac T B\right)T\log(K)}\right).
\]
\end{lemma}
\begin{proof}
The proof is via a reduction to the setting of sequential adversarial multi-armed bandits with delays, in which the reward received in each round is revealed to the player $D$ rounds later.
For this problem, 
\citet[Corollary~15]{delayed_bandits}
gave an algorithm with regret
\[
O\left( D + \sqrt{(K+D)T\log(K)}\right).
\]

Now, for the batched adversarial bandit problem, we partition the time horizon into $B$ batches of size $T/B$.
Thus, the reward of each pull is revealed at most $T/B$ rounds later, hence the above result gives an algorithm with regret
\begin{align*}
\mathbf{E}[
\textnormal{Regret}] \leq &O\left( \frac T B + \sqrt{\left(K+\frac T B\right)T\log(K)}\right)
\\=&
O\left(\sqrt{\left(K+\frac T B\right)T\log(K)}\right),
\end{align*}
as required.
\end{proof}


We complement the above result with a nearly tight lower bound.

\begin{lemma}
\label{lem:adversariallowerbound2}
Any non-adaptive algorithm for batched adversarial multi-armed bandits has regret at least $\Omega\left(\frac T {\sqrt{B}}\right)$.
\end{lemma}

\begin{proof}
Suppose the batch sizes are $t_1,\dots,t_B$, which are fixed before starting the game.
Consider the following example with $K=2$ arms.
For each batch, we choose a uniformly random arm and set its reward to $1$ throughout the batch, and set the other arm's reward to $0$.
Since the expected instantaneous reward of any round given information from past is $\frac 1 2$, the expected reward of \emph{any} non-adaptive algorithm is $\frac T 2$.

Next we show that the reward of one of the arms is $\frac T 2 + \Omega\left(\frac T {\sqrt B}\right)$.
Let $X \coloneqq t_1 R_1 + \dots + t_B R_B$, where each $R_i$ is $-1$ or $+1$ independently and uniformly at random.
Then the total rewards of the two arms in our instance are distributed as $\frac T 2 - \frac X 2$ and $\frac T 2 + \frac X 2$. To complete the proof, we need only show that $E[|X|]=\Omega( T/{\sqrt {B}})$.
We will use H\"older's inequality,
\[
\E[|f(X)|^p]^{1/p} \cdot
\E[|g(X)|^q]^{1/q}
\geq
\E [|f(X) g(X)|],
\]
valid for any $0<p,q<1$ with $1/p+1/q=1$, with
$f(x)\coloneqq x^{4/3}$,
$g(x)\coloneqq x^{2/3}$,
$p=3,q=3/2$, which yields
\begin{equation}
\label{eq:holder}
\left(\E|X|^4\right)^{1/3}
\left(\E|X|\right)^{2/3}
\geq
\E|X|^2.
\end{equation}
Thus, to lower bound $\E|X|$ we need to bound  
$\E|X|^2$
and
$\E|X|^4$.
For $\E|X|^2=\E X^2$, observe that
\begin{align*}
\E [X^2] & = 
\E \left[\left(\sum_{i=1}^{B} t_i R_i\right)^2\right] \\
& = \E \left[\sum_{i=1}^{B} t_i^2 R_i^2\right]
+ \E \left[\sum_{i\neq j} t_i t_j R_i R_j\right] \\
& = 
\E \sum_{i=1}^{B} t_i^2 R_i^2
= \sum_{i=1}^{B} t_i^2,
\end{align*}
since $\E R_i=0, \E R_i^2=1$, and $R_i$ and $R_j$ are independent for $i\neq j$.
Similarly, after expanding
$X^4 = (\sum t_i R_i)^4$,
all terms with odd powers of $R_i$ will have zero expectations, 
and only terms of type $t_i^4 R_i^4$ and $t_i^2 t_j^2 R_i^2 R_j^2$ remain after taking expectation, so we get
\begin{align*}
\E [X^4]   = 
\sum_{i=1}^{B} t_i^4 + 
6 \sum_{i<j} t_i^2t_j^2 \leq 
 3  \left(\sum_{i=1}^{B} t_i^2\right)^2.
\end{align*}
From~\eqref{eq:holder} we get
\[
\E [|X|] \geq \frac{\big(\E [X^2]\big)^{\frac 3 2}}{\big(\E [X^4]\big)^{\frac 1 2}} \geq \sqrt{\frac 1 3 \sum_{i=1}^B t_i^2}
\geq \frac T {\sqrt{3B}},\]
where the last inequality follows from the Cauchy-Schwarz inequality, recalling that $\sum_{i=1}^B t_i=T$.
Hence, the expected regret is at least $\E|X|/2=\Omega( T /\sqrt{B})$, completing the proof.
\end{proof}

We are now ready to prove the main result of this section, which is a minimax regret characterization of non-adaptive algorithms for batched adversarial multi-armed bandits.

\begin{theorem}
\label{thm:abmain}
The best achievable regret of a non-adaptive algorithm for batched adversarial multi-armed bandits is
\[
\widetilde{\Theta}
\left(
\sqrt{T \left(K + \frac T B\right)}
\right).
\]
\end{theorem}
\begin{proof}
An upper bound of
\(
O\left(
\sqrt{T \left(K + \frac T B\right)\log(K)}
\right)
\)
is proved in Lemma~\ref{lem:adversarialupperbound}.
A lower bound of $\Omega({T}/{\sqrt{B}})$ is proved in
Lemma~\ref{lem:adversariallowerbound2}, while a lower bound of $\Omega(\sqrt{KT})$ holds even in the sequential setting when $B=T$
\citep[Theorem~5.1]{adversarialbandits}.
\end{proof}

Finally, for adaptive algorithms for batched adversarial multi-armed bandits, we show a regret lower bound of $\Omega ( T / B)$.

\begin{theorem}
\label{thm:adversariallowerbound1}
Any adaptive algorithm for batched adversarial multi-armed bandits has regret at least $\Omega( T / B)$.
\end{theorem}
\begin{proof}
We first prove the lower bound for non-adaptive algorithms and then extend it to adaptive algorithms.

Let $K=2$ and consider the following reward sequences. In the beginning, both arms have $0$ rewards. Then, at a round chosen uniformly at random from $\{1,\dots,T\}$, the reward of one of the arms becomes $1$ and stays $1$ until the end. Hence, the expected reward of the best arm is $\frac T 2$.

The switching happens inside one of the batches. The expected number of $1$s that fall in that batch is half of the size of the batch, and for any strategy chosen by the player in that batch, her expected regret is at least a quarter of the size of the batch.

Denote the batch sizes by $t_1,\dots,t_B$. The probability that the (random) switching time falls in the $i$th batch is $\frac {t_i}{T}$. Hence, the expected regret is at least
\[
\sum_{i=1}^{B} \frac{t_i}{T} \cdot \frac{t_i} 4 =
\frac{1}{4 T} \cdot \sum_{i=1}^{B} t_i^2
\geq \frac {1}{4 T} \cdot B \cdot \Big(\frac T B \Big)^2
= \frac {T}{4B},
\]
completing the proof.

To extend the lower bound to adaptive algorithms, note that the defined distribution over reward sequences does not depend on the batch sizes or the algorithms' actions.
Hence, the lower bound holds for \emph{any} sequence of batch sizes, deterministic or randomized.
\end{proof}



\section{Conclusion}
\label{sec:conclusion}
This paper presented a systematic theoretical study of the batched bandits problem in stochastic and adversarial settings.
We have shown a large contrast between the stochastic and adversarial multi-armed bandits: while in the stochastic case
a logarithmic number of batches are enough to achieve the optimal regret, the adversarial case needs a polynomial number of batches.
This motivates studying batched versions of models in-between stochastic and adversarial; one such model is the non-stationary model.

Starting from the stochastic model, a \emph{non-stationary} multi-armed bandit problem is one in which the arms reward distributions may change over time, but there is a restriction on the amount of change.
A natural assumption is to bound the {\em number of changes} in the arms' reward distributions. 
Let $S$ denote the allowed number of changes (or \emph{switches}) of the vector of reward distributions during the $T$ rounds of the game.
The case $S=0$ corresponds to stochastic bandits, while $S=T$ corresponds to adversarial bandits.
This problem has been studied in the sequential setting and various algorithms have been devised based on, e.g., UCB~\citep{switchingbandits} and EXP3~\citep{adversarialbandits}.
It is only natural to study the regret of non-stationary bandits in the batch mode; in particular, the construction of Theorem~\ref{thm:adversariallowerbound1} gives a regret lower bound of $\Omega(T/B)$ 
for any $S>0$; proving regret bounds for all $S$ is an interesting avenue for further research.

For batched stochastic multi-armed bandits with two arms,
\citet[Theorem~2]{batchedtwoarms} provide a regret bound of 
\( O\left(\left(\frac{T}{\log T}\right)^{1/B} \frac{\log T}{\Delta}\right)
\).
It is natural to ask whether this bound can be extended to the case $K>2$: is there an algorithm with regret bounded by
\[
O\left(\left(\frac{T}{\log T}\right)^{1/B}  \sum_{i:\Delta_i>0}\frac{\log(T)}{\Delta_i}
\right)?
\]




\bibliographystyle{plainnat}
\bibliography{batchedmab}

\end{document}